\documentclass[11pt]{article}
\usepackage[a4paper]{geometry}
\usepackage{amsfonts, amsmath, amssymb, amsthm, array, graphicx, caption, authblk, multirow, makecell, framed, float, xcolor, enumitem, tikz, hyperref}
\definecolor{grn}{RGB}{0,127,0}
\definecolor{blk}{RGB}{63,63,63}
\theoremstyle{definition}

\newtheorem{lemma}{Lemma}

\theoremstyle{remark}

\newcommand*{\mybox}[1]{%
  \framebox{\raisebox{0cm}[0.5\baselineskip][0.05\baselineskip]{%
    \hbox to 0.1cm {\hss#1\hss}}}\hspace{0.05cm}}
\newcommand{\mystack}[1]{\begin{array}{|r|||}\hline#1\\\hline\end{array}}

\begin{document}
\title{An Improved Physical ZKP for Nonogram and Nonogram Color}
\author[1]{Suthee Ruangwises\thanks{\texttt{ruangwises@gmail.com}}}
\affil[1]{Department of Informatics, The University of Electro-Communications, Tokyo, Japan}
\date{}
\maketitle

\begin{abstract}
Nonogram is a pencil puzzle consisting of a rectangular white grid where the player has to paint some cells black according to given constraints. In 2010, Chien and Hon constructed a physical card-based zero-knowledge proof protocol for Nonogram, which enables a prover to physically show that he/she knows a solution of the puzzle without revealing it. However, their protocol requires special tools such as scratch-off cards and a sealing machine, making it impractical to implement in real world. The protocol also has a nonzero soundness error. In this paper, we develop a more practical card-based protocol for Nonogram with perfect soundness that uses only regular paper cards. We also show how to modify our protocol to make it support Nonogram Color, a generalization of Nonogram where the player has to paint the cells with multiple colors.

\end{abstract}

\section{Introduction}
\textit{Nonogram} (also known as \textit{Picross}, \textit{Pic-a-Pix}, \textit{Griddlers}, or \textit{Hanjie}) is one of the world's most popular pencil puzzles alongside Sudoku, Numberlink, and other puzzles. Recently, many Nonogram mobile apps with various names have been developed \cite{google}.

A Nonogram puzzle consists of a rectangular white grid of size $m \times n$. The player has to paint some cells black according to the sequences of positive integers assigned to all rows and columns. Suppose a sequence $(x_1,x_2,\dots,x_k)$ is assigned to a row (resp. column), then that row (resp. column) must contain exactly $k$ blocks of consecutive black cells with lengths $x_1,x_2,\dots,x_k$ in this order from left to right (resp. from top to bottom), with at least one white cell separating adjacent blocks. For example, in Fig. \ref{fig0}, the leftmost column has a sequence $(5,2)$ assigned to it, so it must contain a block of five consecutive black cells, followed by a block of three consecutive black cells to the bottom of it, separated by at least one white cell.

\textit{Nonogram Color}, or \textit{Multicolor Nonogram}, is a generalization of Nonogram. It is also a popular puzzle with many mobile apps \cite{google2}. In Nonogram Color, the given numbers have several colors instead of just black, and the block of $x_i$ consecutive cells corresponding to the number $x_i$ must have the same color as the number $x_i$.

It is important to note that the key difference from Nonogram is that in Nonogram Color, blocks with different colors \textit{can} touch, i.e. adjacent blocks with different colors can be right next to each other, while adjacent blocks with the same color must still be separated by at least one white cell like in Nonogram. For example, in Fig. \ref{fig00}, the fifth topmost row has a sequence $(3,3,1)$ (with colors red, green, and green, respectively) assigned to it, so it must contain a block of three consecutive red cells, followed by a block of three consecutive green cells to the right of it, and then a block of one green cell to the right of it; the first and second blocks can be right next to each other, but the second and third blocks must be separated by at least one white cell.

Determining whether a given Nonogram puzzle has a solution has been proved to be NP-complete \cite{np}. As Nonogram Color is a generalization of Nonogram, it is also an NP-complete problem.

Suppose Patricia, an expert in Nonogram, constructed a difficult Nonogram puzzle and challenged her friend Victor to solve it. After several tries, Victor could not solve her puzzle and doubted whether it actually has a solution. Patricia wants to convince him that her puzzle indeed has a solution without revealing it to him, as this would render the challenge pointless. In this situation, Patricia needs some kind of \textit{zero-knowledge proof (ZKP)} protocol.

\begin{figure}
\centering
\begin{tikzpicture}
\draw[step=0.4cm] (0,0) grid (4,4);

\draw[line width=0.5mm] (-0.9,0) -- (4,0);
\draw[line width=0.5mm] (-0.9,4) -- (4,4);
\draw[line width=0.5mm] (-0.9,0) -- (-0.9,4);
\draw (0,0.4) -- (-0.9,0.4);
\draw (0,0.8) -- (-0.9,0.8);
\draw (0,1.2) -- (-0.9,1.2);
\draw (0,1.6) -- (-0.9,1.6);
\draw (0,2.0) -- (-0.9,2.0);
\draw (0,2.4) -- (-0.9,2.4);
\draw (0,2.8) -- (-0.9,2.8);
\draw (0,3.2) -- (-0.9,3.2);
\draw (0,3.6) -- (-0.9,3.6);
\node at (-0.15,0.2) {2};
\node at (-0.45,0.2) {3};
\node at (-0.15,0.6) {3};
\node at (-0.15,1.0) {1};
\node at (-0.45,1.0) {3};
\node at (-0.15,1.4) {1};
\node at (-0.45,1.4) {2};
\node at (-0.15,1.8) {3};
\node at (-0.45,1.8) {4};
\node at (-0.15,2.2) {1};
\node at (-0.45,2.2) {3};
\node at (-0.15,2.6) {1};
\node at (-0.45,2.6) {2};
\node at (-0.75,2.6) {3};
\node at (-0.15,3.0) {1};
\node at (-0.45,3.0) {2};
\node at (-0.75,3.0) {2};
\node at (-0.15,3.4) {2};
\node at (-0.45,3.4) {4};
\node at (-0.75,3.4) {2};
\node at (-0.15,3.8) {2};
\node at (-0.45,3.8) {3};

\draw[line width=0.5mm] (0,0) -- (0,5.2);
\draw[line width=0.5mm] (4,0) -- (4,5.2);
\draw[line width=0.5mm] (0,5.2) -- (4,5.2);
\draw (0.4,4) -- (0.4,5.2);
\draw (0.8,4) -- (0.8,5.2);
\draw (1.2,4) -- (1.2,5.2);
\draw (1.6,4) -- (1.6,5.2);
\draw (2.0,4) -- (2.0,5.2);
\draw (2.4,4) -- (2.4,5.2);
\draw (2.8,4) -- (2.8,5.2);
\draw (3.2,4) -- (3.2,5.2);
\draw (3.6,4) -- (3.6,5.2);
\node at (0.2,4.2) {3};
\node at (0.2,4.6) {5};
\node at (0.6,4.2) {3};
\node at (0.6,4.6) {5};
\node at (1.0,4.2) {7};
\node at (1.4,4.2) {2};
\node at (1.4,4.6) {1};
\node at (1.8,4.2) {2};
\node at (2.2,4.2) {4};
\node at (2.6,4.2) {4};
\node at (3.0,4.2) {1};
\node at (3.0,4.6) {1};
\node at (3.0,5) {1};
\node at (3.4,4.2) {1};
\node at (3.4,4.6) {1};
\node at (3.4,5) {2};
\node at (3.8,4.2) {7};
\end{tikzpicture}
\hspace{1.2cm}
\begin{tikzpicture}
\filldraw[draw=blk,fill=blk] (0,0) rectangle (1.2,1.2);
\filldraw[draw=blk,fill=blk] (0.8,1.2) rectangle (1.6,2);
\filldraw[draw=blk,fill=blk] (0,1.6) rectangle (1.2,2.8);
\filldraw[draw=blk,fill=blk] (0,2.8) rectangle (0.8,3.6);
\filldraw[draw=blk,fill=blk] (2,2.4) rectangle (2.8,4);
\filldraw[draw=blk,fill=blk] (1.6,3.2) rectangle (2,4);
\filldraw[draw=blk,fill=blk] (1.2,3.2) rectangle (1.6,3.6);
\filldraw[draw=blk,fill=blk] (2.8,0) rectangle (3.6,0.4);
\filldraw[draw=blk,fill=blk] (2.8,0.8) rectangle (3.2,1.2);
\filldraw[draw=blk,fill=blk] (3.6,1.2) rectangle (4,4);
\filldraw[draw=blk,fill=blk] (2.8,1.6) rectangle (3.6,2);
\filldraw[draw=blk,fill=blk] (3.2,3.2) rectangle (3.6,4);

\draw[step=0.4cm] (0,0) grid (4,4);

\draw[line width=0.5mm] (-0.9,0) -- (4,0);
\draw[line width=0.5mm] (-0.9,4) -- (4,4);
\draw[line width=0.5mm] (-0.9,0) -- (-0.9,4);
\draw (0,0.4) -- (-0.9,0.4);
\draw (0,0.8) -- (-0.9,0.8);
\draw (0,1.2) -- (-0.9,1.2);
\draw (0,1.6) -- (-0.9,1.6);
\draw (0,2.0) -- (-0.9,2.0);
\draw (0,2.4) -- (-0.9,2.4);
\draw (0,2.8) -- (-0.9,2.8);
\draw (0,3.2) -- (-0.9,3.2);
\draw (0,3.6) -- (-0.9,3.6);
\node at (-0.15,0.2) {2};
\node at (-0.45,0.2) {3};
\node at (-0.15,0.6) {3};
\node at (-0.15,1.0) {1};
\node at (-0.45,1.0) {3};
\node at (-0.15,1.4) {1};
\node at (-0.45,1.4) {2};
\node at (-0.15,1.8) {3};
\node at (-0.45,1.8) {4};
\node at (-0.15,2.2) {1};
\node at (-0.45,2.2) {3};
\node at (-0.15,2.6) {1};
\node at (-0.45,2.6) {2};
\node at (-0.75,2.6) {3};
\node at (-0.15,3.0) {1};
\node at (-0.45,3.0) {2};
\node at (-0.75,3.0) {2};
\node at (-0.15,3.4) {2};
\node at (-0.45,3.4) {4};
\node at (-0.75,3.4) {2};
\node at (-0.15,3.8) {2};
\node at (-0.45,3.8) {3};

\draw[line width=0.5mm] (0,0) -- (0,5.2);
\draw[line width=0.5mm] (4,0) -- (4,5.2);
\draw[line width=0.5mm] (0,5.2) -- (4,5.2);
\draw (0.4,4) -- (0.4,5.2);
\draw (0.8,4) -- (0.8,5.2);
\draw (1.2,4) -- (1.2,5.2);
\draw (1.6,4) -- (1.6,5.2);
\draw (2.0,4) -- (2.0,5.2);
\draw (2.4,4) -- (2.4,5.2);
\draw (2.8,4) -- (2.8,5.2);
\draw (3.2,4) -- (3.2,5.2);
\draw (3.6,4) -- (3.6,5.2);
\node at (0.2,4.2) {3};
\node at (0.2,4.6) {5};
\node at (0.6,4.2) {3};
\node at (0.6,4.6) {5};
\node at (1.0,4.2) {7};
\node at (1.4,4.2) {2};
\node at (1.4,4.6) {1};
\node at (1.8,4.2) {2};
\node at (2.2,4.2) {4};
\node at (2.6,4.2) {4};
\node at (3.0,4.2) {1};
\node at (3.0,4.6) {1};
\node at (3.0,5) {1};
\node at (3.4,4.2) {1};
\node at (3.4,4.6) {1};
\node at (3.4,5) {2};
\node at (3.8,4.2) {7};
\end{tikzpicture}
\caption{An example of a Nonogram puzzle (left) and its solution (right)}
\label{fig0}
\end{figure}

\begin{figure}
\centering
\begin{tikzpicture}
\draw[step=0.4cm] (0,0) grid (4,4);

\draw[line width=0.5mm] (-1.2,0) -- (4,0);
\draw[line width=0.5mm] (-1.2,4) -- (4,4);
\draw[line width=0.5mm] (-1.2,0) -- (-1.2,4);
\draw (0,0.4) -- (-1.2,0.4);
\draw (0,0.8) -- (-1.2,0.8);
\draw (0,1.2) -- (-1.2,1.2);
\draw (0,1.6) -- (-1.2,1.6);
\draw (0,2.0) -- (-1.2,2.0);
\draw (0,2.4) -- (-1.2,2.4);
\draw (0,2.8) -- (-1.2,2.8);
\draw (0,3.2) -- (-1.2,3.2);
\draw (0,3.6) -- (-1.2,3.6);
\node at (-0.15,0.2) {\textcolor{blue}{\textbf{1}}};
\node at (-0.45,0.2) {\textcolor{blue}{\textbf{1}}};
\node at (-0.75,0.2) {\textcolor{red}{4}};
\node at (-0.15,0.6) {\textcolor{blue}{\textbf{6}}};
\node at (-0.45,0.6) {\textcolor{red}{2}};
\node at (-0.15,1.0) {\textcolor{blue}{\textbf{1}}};
\node at (-0.45,1.0) {\textcolor{blue}{\textbf{1}}};
\node at (-0.75,1.0) {\textcolor{red}{1}};
\node at (-0.15,1.4) {\textcolor{blue}{\textbf{4}}};
\node at (-0.45,1.4) {\textcolor{red}{2}};
\node at (-0.15,1.8) {\textcolor{grn}{\textit{1}}};
\node at (-0.45,1.8) {\textcolor{red}{1}};
\node at (-0.75,1.8) {\textcolor{grn}{\textit{2}}};
\node at (-0.15,2.2) {\textcolor{grn}{\textit{1}}};
\node at (-0.45,2.2) {\textcolor{grn}{\textit{3}}};
\node at (-0.75,2.2) {\textcolor{red}{3}};
\node at (-0.15,2.6) {\textcolor{grn}{\textit{3}}};
\node at (-0.45,2.6) {\textcolor{red}{2}};
\node at (-0.75,2.6) {\textcolor{grn}{\textit{1}}};
\node at (-1.05,2.6) {\textcolor{red}{1}};
\node at (-0.15,3.0) {\textcolor{grn}{\textit{2}}};
\node at (-0.45,3.0) {\textcolor{grn}{\textit{1}}};
\node at (-0.75,3.0) {\textcolor{red}{2}};
\node at (-1.05,3.0) {\textcolor{grn}{\textit{3}}};
\node at (-0.15,3.4) {\textcolor{grn}{\textit{1}}};
\node at (-0.45,3.4) {\textcolor{grn}{\textit{3}}};
\node at (-0.75,3.4) {\textcolor{red}{1}};
\node at (-1.05,3.4) {\textcolor{grn}{\textit{2}}};
\node at (-0.15,3.8) {\textcolor{grn}{\textit{7}}};

\draw[line width=0.5mm] (0,0) -- (0,5.6);
\draw[line width=0.5mm] (4,0) -- (4,5.6);
\draw[line width=0.5mm] (0,5.6) -- (4,5.6);
\draw (0.4,4) -- (0.4,5.6);
\draw (0.8,4) -- (0.8,5.6);
\draw (1.2,4) -- (1.2,5.6);
\draw (1.6,4) -- (1.6,5.6);
\draw (2.0,4) -- (2.0,5.6);
\draw (2.4,4) -- (2.4,5.6);
\draw (2.8,4) -- (2.8,5.6);
\draw (3.2,4) -- (3.2,5.6);
\draw (3.6,4) -- (3.6,5.6);
\node at (0.2,4.2) {\textcolor{red}{1}};
\node at (0.2,4.6) {\textcolor{grn}{\textit{1}}};
\node at (0.2,5) {\textcolor{red}{2}};
\node at (0.2,5.4) {\textcolor{grn}{\textit{1}}};
\node at (0.6,4.2) {\textcolor{red}{4}};
\node at (0.6,4.6) {\textcolor{grn}{\textit{1}}};
\node at (0.6,5) {\textcolor{red}{1}};
\node at (0.6,5.4) {\textcolor{grn}{\textit{4}}};
\node at (1.0,4.2) {\textcolor{red}{2}};
\node at (1.0,4.6) {\textcolor{red}{4}};
\node at (1.0,5) {\textcolor{grn}{\textit{3}}};
\node at (1.4,4.2) {\textcolor{red}{1}};
\node at (1.4,4.6) {\textcolor{grn}{\textit{2}}};
\node at (1.4,5) {\textcolor{red}{3}};
\node at (1.4,5.4) {\textcolor{grn}{\textit{1}}};
\node at (1.8,4.2) {\textcolor{blue}{\textbf{2}}};
\node at (1.8,4.6) {\textcolor{grn}{\textit{2}}};
\node at (1.8,5) {\textcolor{red}{1}};
\node at (1.8,5.4) {\textcolor{grn}{\textit{2}}};
\node at (2.2,4.2) {\textcolor{blue}{\textbf{3}}};
\node at (2.2,4.6) {\textcolor{grn}{\textit{5}}};
\node at (2.6,4.2) {\textcolor{blue}{\textbf{1}}};
\node at (2.6,4.6) {\textcolor{blue}{\textbf{1}}};
\node at (2.6,5) {\textcolor{grn}{\textit{1}}};
\node at (2.6,5.4) {\textcolor{grn}{\textit{2}}};
\node at (3.0,4.2) {\textcolor{blue}{\textbf{1}}};
\node at (3.0,4.6) {\textcolor{blue}{\textbf{1}}};
\node at (3.0,5) {\textcolor{grn}{\textit{1}}};
\node at (3.4,4.2) {\textcolor{blue}{\textbf{3}}};
\node at (3.4,4.6) {\textcolor{grn}{\textit{2}}};
\node at (3.8,4.2) {\textcolor{blue}{\textbf{2}}};
\node at (3.8,4.6) {\textcolor{grn}{\textit{1}}};
\node at (3.8,5) {\textcolor{grn}{\textit{1}}};
\end{tikzpicture}
\hspace{0.6cm}
\begin{tikzpicture}
\filldraw[draw=red,fill=red] (0,0) rectangle (1.6,0.4);
\filldraw[draw=red,fill=red] (0.4,0.4) rectangle (1.2,0.8);
\filldraw[draw=red,fill=red] (0.4,0.8) rectangle (0.8,1.6);
\filldraw[draw=red,fill=red] (0.8,1.2) rectangle (1.2,2.8);
\filldraw[draw=red,fill=red] (1.2,2.4) rectangle (1.6,3.6);
\filldraw[draw=red,fill=red] (1.6,2.8) rectangle (2,3.2);
\filldraw[draw=red,fill=red] (0,2) rectangle (0.8,2.4);
\filldraw[draw=red,fill=red] (0,2.4) rectangle (0.4,2.8);

\filldraw[draw=grn,fill=grn] (0,1.6) rectangle (0.8,2);
\filldraw[draw=grn,fill=grn] (1.2,1.6) rectangle (1.6,2.4);
\filldraw[draw=grn,fill=grn] (1.6,2) rectangle (2.4,2.8);
\filldraw[draw=grn,fill=grn] (2.4,2.4) rectangle (2.8,2.8);
\filldraw[draw=grn,fill=grn] (2,2.8) rectangle (2.4,3.2);
\filldraw[draw=grn,fill=grn] (1.6,3.2) rectangle (2.8,3.6);
\filldraw[draw=grn,fill=grn] (0.4,3.6) rectangle (3.2,4);
\filldraw[draw=grn,fill=grn] (0.4,2.8) rectangle (1.2,3.6);
\filldraw[draw=grn,fill=grn] (0,2.8) rectangle (0.4,3.2);
\filldraw[draw=grn,fill=grn] (0.4,2.4) rectangle (0.8,2.8);
\filldraw[draw=grn,fill=grn] (3.2,2.8) rectangle (4,3.2);
\filldraw[draw=grn,fill=grn] (3.2,3.2) rectangle (3.6,3.6);
\filldraw[draw=grn,fill=grn] (3.6,2) rectangle (4,2.4);

\filldraw[draw=blue,fill=blue] (1.6,0) rectangle (2,0.4);
\filldraw[draw=blue,fill=blue] (3.6,0) rectangle (4,0.4);
\filldraw[draw=blue,fill=blue] (1.6,0.4) rectangle (4,0.8);
\filldraw[draw=blue,fill=blue] (2,0.8) rectangle (2.4,1.2);
\filldraw[draw=blue,fill=blue] (3.2,0.8) rectangle (3.6,1.2);
\filldraw[draw=blue,fill=blue] (2,1.2) rectangle (3.6,1.6);

\draw[step=0.4cm] (0,0) grid (4,4);

\draw[line width=0.5mm] (-1.2,0) -- (4,0);
\draw[line width=0.5mm] (-1.2,4) -- (4,4);
\draw[line width=0.5mm] (-1.2,0) -- (-1.2,4);
\draw (0,0.4) -- (-1.2,0.4);
\draw (0,0.8) -- (-1.2,0.8);
\draw (0,1.2) -- (-1.2,1.2);
\draw (0,1.6) -- (-1.2,1.6);
\draw (0,2.0) -- (-1.2,2.0);
\draw (0,2.4) -- (-1.2,2.4);
\draw (0,2.8) -- (-1.2,2.8);
\draw (0,3.2) -- (-1.2,3.2);
\draw (0,3.6) -- (-1.2,3.6);
\node at (-0.15,0.2) {\textcolor{blue}{\textbf{1}}};
\node at (-0.45,0.2) {\textcolor{blue}{\textbf{1}}};
\node at (-0.75,0.2) {\textcolor{red}{4}};
\node at (-0.15,0.6) {\textcolor{blue}{\textbf{6}}};
\node at (-0.45,0.6) {\textcolor{red}{2}};
\node at (-0.15,1.0) {\textcolor{blue}{\textbf{1}}};
\node at (-0.45,1.0) {\textcolor{blue}{\textbf{1}}};
\node at (-0.75,1.0) {\textcolor{red}{1}};
\node at (-0.15,1.4) {\textcolor{blue}{\textbf{4}}};
\node at (-0.45,1.4) {\textcolor{red}{2}};
\node at (-0.15,1.8) {\textcolor{grn}{\textit{1}}};
\node at (-0.45,1.8) {\textcolor{red}{1}};
\node at (-0.75,1.8) {\textcolor{grn}{\textit{2}}};
\node at (-0.15,2.2) {\textcolor{grn}{\textit{1}}};
\node at (-0.45,2.2) {\textcolor{grn}{\textit{3}}};
\node at (-0.75,2.2) {\textcolor{red}{3}};
\node at (-0.15,2.6) {\textcolor{grn}{\textit{3}}};
\node at (-0.45,2.6) {\textcolor{red}{2}};
\node at (-0.75,2.6) {\textcolor{grn}{\textit{1}}};
\node at (-1.05,2.6) {\textcolor{red}{1}};
\node at (-0.15,3.0) {\textcolor{grn}{\textit{2}}};
\node at (-0.45,3.0) {\textcolor{grn}{\textit{1}}};
\node at (-0.75,3.0) {\textcolor{red}{2}};
\node at (-1.05,3.0) {\textcolor{grn}{\textit{3}}};
\node at (-0.15,3.4) {\textcolor{grn}{\textit{1}}};
\node at (-0.45,3.4) {\textcolor{grn}{\textit{3}}};
\node at (-0.75,3.4) {\textcolor{red}{1}};
\node at (-1.05,3.4) {\textcolor{grn}{\textit{2}}};
\node at (-0.15,3.8) {\textcolor{grn}{\textit{7}}};

\draw[line width=0.5mm] (0,0) -- (0,5.6);
\draw[line width=0.5mm] (4,0) -- (4,5.6);
\draw[line width=0.5mm] (0,5.6) -- (4,5.6);
\draw (0.4,4) -- (0.4,5.6);
\draw (0.8,4) -- (0.8,5.6);
\draw (1.2,4) -- (1.2,5.6);
\draw (1.6,4) -- (1.6,5.6);
\draw (2.0,4) -- (2.0,5.6);
\draw (2.4,4) -- (2.4,5.6);
\draw (2.8,4) -- (2.8,5.6);
\draw (3.2,4) -- (3.2,5.6);
\draw (3.6,4) -- (3.6,5.6);
\node at (0.2,4.2) {\textcolor{red}{1}};
\node at (0.2,4.6) {\textcolor{grn}{\textit{1}}};
\node at (0.2,5) {\textcolor{red}{2}};
\node at (0.2,5.4) {\textcolor{grn}{\textit{1}}};
\node at (0.6,4.2) {\textcolor{red}{4}};
\node at (0.6,4.6) {\textcolor{grn}{\textit{1}}};
\node at (0.6,5) {\textcolor{red}{1}};
\node at (0.6,5.4) {\textcolor{grn}{\textit{4}}};
\node at (1.0,4.2) {\textcolor{red}{2}};
\node at (1.0,4.6) {\textcolor{red}{4}};
\node at (1.0,5) {\textcolor{grn}{\textit{3}}};
\node at (1.4,4.2) {\textcolor{red}{1}};
\node at (1.4,4.6) {\textcolor{grn}{\textit{2}}};
\node at (1.4,5) {\textcolor{red}{3}};
\node at (1.4,5.4) {\textcolor{grn}{\textit{1}}};
\node at (1.8,4.2) {\textcolor{blue}{\textbf{2}}};
\node at (1.8,4.6) {\textcolor{grn}{\textit{2}}};
\node at (1.8,5) {\textcolor{red}{1}};
\node at (1.8,5.4) {\textcolor{grn}{\textit{2}}};
\node at (2.2,4.2) {\textcolor{blue}{\textbf{3}}};
\node at (2.2,4.6) {\textcolor{grn}{\textit{5}}};
\node at (2.6,4.2) {\textcolor{blue}{\textbf{1}}};
\node at (2.6,4.6) {\textcolor{blue}{\textbf{1}}};
\node at (2.6,5) {\textcolor{grn}{\textit{1}}};
\node at (2.6,5.4) {\textcolor{grn}{\textit{2}}};
\node at (3.0,4.2) {\textcolor{blue}{\textbf{1}}};
\node at (3.0,4.6) {\textcolor{blue}{\textbf{1}}};
\node at (3.0,5) {\textcolor{grn}{\textit{1}}};
\node at (3.4,4.2) {\textcolor{blue}{\textbf{3}}};
\node at (3.4,4.6) {\textcolor{grn}{\textit{2}}};
\node at (3.8,4.2) {\textcolor{blue}{\textbf{2}}};
\node at (3.8,4.6) {\textcolor{grn}{\textit{1}}};
\node at (3.8,5) {\textcolor{grn}{\textit{1}}};

\node at (0.2,0.2) {\textcolor{white}{R}};
\node at (0.6,0.2) {\textcolor{white}{R}};
\node at (1,0.2) {\textcolor{white}{R}};
\node at (1.4,0.2) {\textcolor{white}{R}};
\node at (0.6,0.6) {\textcolor{white}{R}};
\node at (1,0.6) {\textcolor{white}{R}};
\node at (0.6,1) {\textcolor{white}{R}};
\node at (0.6,1.4) {\textcolor{white}{R}};
\node at (1,1.4) {\textcolor{white}{R}};
\node at (1,1.8) {\textcolor{white}{R}};
\node at (0.2,2.2) {\textcolor{white}{R}};
\node at (0.6,2.2) {\textcolor{white}{R}};
\node at (1,2.2) {\textcolor{white}{R}};
\node at (0.2,2.6) {\textcolor{white}{R}};
\node at (1,2.6) {\textcolor{white}{R}};
\node at (1.4,2.6) {\textcolor{white}{R}};
\node at (1.4,3) {\textcolor{white}{R}};
\node at (1.8,3) {\textcolor{white}{R}};
\node at (1.4,3.4) {\textcolor{white}{R}};

\node at (0.2,1.8) {\textcolor{white}{G}};
\node at (0.6,1.8) {\textcolor{white}{G}};
\node at (1.4,1.8) {\textcolor{white}{G}};
\node at (1.4,2.2) {\textcolor{white}{G}};
\node at (1.8,2.2) {\textcolor{white}{G}};
\node at (2.2,2.2) {\textcolor{white}{G}};
\node at (3.8,2.2) {\textcolor{white}{G}};
\node at (0.6,2.6) {\textcolor{white}{G}};
\node at (1.8,2.6) {\textcolor{white}{G}};
\node at (2.2,2.6) {\textcolor{white}{G}};
\node at (2.6,2.6) {\textcolor{white}{G}};
\node at (0.2,3) {\textcolor{white}{G}};
\node at (0.6,3) {\textcolor{white}{G}};
\node at (1,3) {\textcolor{white}{G}};
\node at (2.2,3) {\textcolor{white}{G}};
\node at (3.4,3) {\textcolor{white}{G}};
\node at (3.8,3) {\textcolor{white}{G}};
\node at (0.6,3.4) {\textcolor{white}{G}};
\node at (1,3.4) {\textcolor{white}{G}};
\node at (1.8,3.4) {\textcolor{white}{G}};
\node at (2.2,3.4) {\textcolor{white}{G}};
\node at (2.6,3.4) {\textcolor{white}{G}};
\node at (3.4,3.4) {\textcolor{white}{G}};
\node at (0.6,3.8) {\textcolor{white}{G}};
\node at (1,3.8) {\textcolor{white}{G}};
\node at (1.4,3.8) {\textcolor{white}{G}};
\node at (1.8,3.8) {\textcolor{white}{G}};
\node at (2.2,3.8) {\textcolor{white}{G}};
\node at (2.6,3.8) {\textcolor{white}{G}};
\node at (3,3.8) {\textcolor{white}{G}};

\node at (1.8,0.2) {\textcolor{white}{B}};
\node at (3.8,0.2) {\textcolor{white}{B}};
\node at (1.8,0.6) {\textcolor{white}{B}};
\node at (2.2,0.6) {\textcolor{white}{B}};
\node at (2.6,0.6) {\textcolor{white}{B}};
\node at (3,0.6) {\textcolor{white}{B}};
\node at (3.4,0.6) {\textcolor{white}{B}};
\node at (3.8,0.6) {\textcolor{white}{B}};
\node at (2.2,1) {\textcolor{white}{B}};
\node at (3.4,1) {\textcolor{white}{B}};
\node at (2.2,1.4) {\textcolor{white}{B}};
\node at (2.6,1.4) {\textcolor{white}{B}};
\node at (3,1.4) {\textcolor{white}{B}};
\node at (3.4,1.4) {\textcolor{white}{B}};
\end{tikzpicture}
\caption{An example of a Nonogram Color puzzle (left) and its solution (right) (Numbers in red, green, and blue are in normal text, italic, and boldfaced, respectively. Cells in red, green, and blue are marked by letters R, G, and B, respectively.)}
\label{fig00}
\end{figure}

\subsection{Zero-Knowledge Proof}
First introduced in 1989 by Goldwasser et al. \cite{zkp0}, a ZKP protocol is an interactive protocol between a prover $P$ and a verifier $V$, where both of them are given an instance $x$ of a computational problem. Only $P$ knows a solution $w$ of $x$, and the computational power of $V$ is so limited that he/she cannot obtain $w$ from $x$. A ZKP protocol enables $P$ to convince $V$ that he/she knows $w$ without revealing any information about $w$ to $V$. Such protocol has to satisfy the following three properties.

\begin{enumerate}
	\item \textbf{Completeness:} If $P$ knows $w$, then $V$ accepts with high probability. (In this paper, we consider only the \textit{perfect completeness} property where $V$ always accepts.)
	\item \textbf{Soundness:} If $P$ does not know $w$, then $V$ always reject, except with a small probability called \textit{soundness error}. (In this paper, we consider only the \textit{perfect soundness} property where the soundness error is zero.)
	\item \textbf{Zero-knowledge:} $V$ cannot obtain any information about $w$, i.e. there exists a probabilistic polynomial time algorithm $S$ (called a \textit{simulator}), not knowing $w$ but having an access to $V$, such that the outputs of $S$ follow the same probability distribution as the ones of the actual protocol.
\end{enumerate}

As there exists a ZKP protocol for every NP problem \cite{zkp}, it is possible to construct a computational ZKP protocol for Nonogram. However, such construction requires cryptographic primitives and thus is not intuitive or practical.

Instead, we aim to develop a physical ZKP protocol for Nonogram using a deck of playing cards. Card-based protocols have benefits that they use only portable objects found in everyday life without requiring computers. Moreover, these protocols are easy to understand and verify the correctness and security, even for non-expert in cryptography. Hence, they can be used for didactic purpose.

\subsection{Related Work}
In 2009, Gradwohl et al. \cite{sudoku0} developed a card-based ZKP protocol for Sudoku, the first of its kind for any pencil puzzle. However, each of their several variants of the protocol either has a nonzero soundness error or requires special tools. Sasaki et al. \cite{sudoku} later constructed a ZKP protocol for Sudoku that achieves perfect soundness without using special tools. Ruangwises \cite{sudoku2} also developed another ZKP protocol for Sudoku that can be implemented using a deck of all different cards with no duplicates.

The second card-based ZKP protocol for a pencil puzzle was the one for Nonogram, developed by Chien and Hon \cite{nonogram} in 2010. Their protocol, however, requires scratch-off cards and a sealing machine, which is difficult to find in everyday life, making it very impractical. Another drawback of their protocol is that it has a nonzero soundness error. In fact, the error is as high as 6/7, which means the protocol has to be repeated for many times until the soundness error becomes reasonably low.

Since then, the area of card-based ZKP protocols has been extensively studied by many researchers. Besides Sudoku and Nonogram, such protocols for many other pencil puzzles have been proposed so far: ABC End View \cite{abc}, Akari \cite{akari}, Bridges \cite{bridges}, Heyawake \cite{nurikabe}, Hitori \cite{nurikabe}, Juosan \cite{takuzu}, Kakuro \cite{akari,kakuro}, KenKen \cite{akari}, Makaro \cite{makaro,makaro2}, Masyu \cite{slitherlink}, Norinori \cite{norinori}, Numberlink \cite{numberlink}, Nurikabe \cite{nurikabe}, Nurimisaki \cite{nurimisaki}, Ripple Effect \cite{ripple}, Shikaku \cite{shikaku}, Slitherlink \cite{slitherlink}, Suguru \cite{suguru}, Takuzu \cite{akari,takuzu}, and Usowan \cite{usowan}. Except for the ones in \cite{akari}, all subsequent protocols have perfect soundness and do not require special tools.

\subsection{Our Contribution}
Although Nonogram is the second pencil puzzle after Sudoku to have a card-based ZKP protocol, it still lacks a protocol with perfect soundness, or a practical one that does not require special tools. The problem of developing either such protocol has remained open for more than ten years.

In this paper, we solve both problems by developing a card-based ZKP protocol for Nonogram with perfect completeness and perfect soundness, using only regular paper cards. Our protocol uses $\Theta(mn)$ cards and $\Theta(mn)$ shuffles in an $m \times n$ Nonogram puzzle.

We also show how to modify our protocol to make it support Nonogram Color as well. Our modified protocol uses $\Theta(mnp)$ cards and $\Theta(mn)$ shuffles in an $m \times n$ Nonogram Color puzzle with $p$ colors (including white).

\section{Preliminaries}
\subsection{Cards}
Each card used in our protocol either has \mybox{$\clubsuit$}, \mybox{$\heartsuit$}, \mybox{$\spadesuit$}, or \mybox{$\diamondsuit$} as front side. All cards have indistinguishable back sides denoted by \mybox{?}.

\subsection{Random Cut} \label{rc}
Given a sequence $S$ of $k$ cards, a \textit{random cut} shifts $S$ by a uniformly random cyclic shift unknown to all parties. It can be implemented by letting all parties take turns to apply a \textit{Hindu cut} (taking several cards from the bottom of the pile and putting them on the top) to $S$ \cite{hindu}.

\subsection{Pile-Shifting Shuffle} \label{pss}
Given an $\ell \times k$ matrix $M$ of cards, a \textit{pile-shifting shuffle} \cite{polygon} shifts the columns of $M$ by a uniformly random cyclic shift unknown to all parties. It can be implemented by putting all cards in each column into an envelope and applying the random cut to the sequence of envelopes.

\subsection{Copy Protocol} \label{copy}
Given an input sequence of two face-down cards, which is either \hbox{\mybox{$\clubsuit$}\mybox{$\heartsuit$}} or \hbox{\mybox{$\heartsuit$}\mybox{$\clubsuit$}}, a \textit{copy protocol} \cite{mizuki} enables $P$ to produce an additional copy of the input sequence without revealing it to $V$. It also verifies to $V$ that the input sequence is indeed either \mybox{$\clubsuit$}\mybox{$\heartsuit$} or \mybox{$\heartsuit$}\mybox{$\clubsuit$} (not \mybox{$\clubsuit$}\mybox{$\clubsuit$} or \mybox{$\heartsuit$}\mybox{$\heartsuit$}).

\begin{figure}[H]
\centering
\begin{tikzpicture}
\node at (0.0,1.4) {\mybox{?}};
\node at (0.6,1.4) {\mybox{?}};

\node at (0.0,0.7) {\mybox{$\clubsuit$}};
\node at (0.6,0.7) {\mybox{$\heartsuit$}};

\node at (0.0,0) {\mybox{$\clubsuit$}};
\node at (0.6,0) {\mybox{$\heartsuit$}};
\end{tikzpicture}
\caption{The matrix $M$ constructed in Step 1 of the copy protocol}
\label{fig2}
\end{figure}

In the copy protocol, $P$ performs the following steps.
\begin{enumerate}
	\item Construct the following $3 \times k$ matrix $M$ (see Fig. \ref{fig2}).
	\begin{enumerate}
		\item In the first row, place the input sequence.
		\item In the second row and third row, publicly place a face-up sequence \hbox{\mybox{$\clubsuit$}\mybox{$\heartsuit$}}.
	\end{enumerate}
	\item Turn over all face-up cards and apply the pile-shifting shuffle to $M$.
	\item Turn over all cards in the first row of $M$. If the revealed sequence is \hbox{\mybox{$\clubsuit$}\mybox{$\heartsuit$}}, do nothing; if the sequence is \hbox{\mybox{$\heartsuit$}\mybox{$\clubsuit$}}, swap the two columns of $M$. (If the sequence is anything else, then $V$ rejects.)
	\item The sequences in the second and third rows of $M$ will be the two copies of the input sequence as desired.
\end{enumerate}

\subsection{Chosen Cut Protocol} \label{chosen}
Given a sequence of $k$ face-down cards $A = (a_1,a_2,\dots,a_k)$, a \textit{chosen cut protocol} \cite{koch} enables $P$ to select a card $a_i$ he/she desires without revealing $i$ to $V$.

\begin{figure}[H]
\centering
\begin{tikzpicture}
\node at (0.0,1.4) {\mybox{?}};
\node at (0.6,1.4) {\mybox{?}};
\node at (1.2,1.4) {\dots};
\node at (1.8,1.4) {\mybox{?}};
\node at (2.4,1.4) {\mybox{?}};
\node at (3.0,1.4) {\mybox{?}};
\node at (3.6,1.4) {\dots};
\node at (4.2,1.4) {\mybox{?}};

\node at (0.0,1) {$a_1$};
\node at (0.6,1) {$a_2$};
\node at (1.8,1) {$a_{i-1}$};
\node at (2.4,1) {$a_i$};
\node at (3.0,1) {$a_{i+1}$};
\node at (4.2,1) {$a_k$};

\node at (0.0,0.4) {\mybox{?}};
\node at (0.6,0.4) {\mybox{?}};
\node at (1.2,0.4) {\dots};
\node at (1.8,0.4) {\mybox{?}};
\node at (2.4,0.4) {\mybox{?}};
\node at (3.0,0.4) {\mybox{?}};
\node at (3.6,0.4) {\dots};
\node at (4.2,0.4) {\mybox{?}};

\node at (0.0,0) {$\heartsuit$};
\node at (0.6,0) {$\heartsuit$};
\node at (1.8,0) {$\heartsuit$};
\node at (2.4,0) {$\clubsuit$};
\node at (3.0,0) {$\heartsuit$};
\node at (4.2,0) {$\heartsuit$};
\end{tikzpicture}
\caption{The matrix $M$ constructed in Step 1 of the chosen cut protocol}
\label{fig3}
\end{figure}

In the chosen cut protocol, $P$ performs the following steps.
\begin{enumerate}
	\item Construct the following $2 \times k$ matrix $M$ (see Fig. \ref{fig3}).
	\begin{enumerate}
		\item In the first row, place the input sequence $A$.
		\item In the second row, secretly place a face-down \mybox{$\clubsuit$} at the $i$-th column, and a face-down \mybox{$\heartsuit$} at each of the rest of columns.
	\end{enumerate}
	\item Apply the pile-shifting shuffle to $M$.
	\item Turn over all cards in the second row of $M$. Locate the position of the only \mybox{$\clubsuit$}. A card in the first row directly above this card will be the card $a_i$ as desired.
\end{enumerate}

\section{Protocol for Nonogram} \label{main}
On each cell in the Nonogram grid, $P$ secretly places a face-down sequence \hbox{\mybox{$\clubsuit$}\mybox{$\heartsuit$}} if the cell is black or \hbox{\mybox{$\heartsuit$}\mybox{$\clubsuit$}} if the cell is white according to $P$'s solution. Then, $P$ publicly applies the copy protocol to the sequence on each cell to produce an additional copy of it. Each of the two copies will be used to verify a row and a column the cell belongs to. Note that the copy protocol also verifies that the sequence on each cell is in a correct format (either \hbox{\mybox{$\clubsuit$}\mybox{$\heartsuit$}} or \hbox{\mybox{$\heartsuit$}\mybox{$\clubsuit$}}).

From now on, we will show the verification of a row $R$ with $n$ cells that has a sequence $(x_1,x_2,\dots,x_k)$ assigned to it. The verification of a column works analogously.

For every cell in $R$, $P$ selects only the left card from the sequence on it (which is a \mybox{$\clubsuit$} if the cell is black and a \mybox{$\heartsuit$} if the cell is white). $P$ then arranges the selected cards as a sequence $S = (a_1,a_2,\dots,a_n)$, with each card in $S$ corresponding to each cell in $R$ in this order from left to right. As $R$ may start and end with a white or black cell, $P$ publicly appends two face-down \mybox{$\heartsuit$}s, called $a_0$ and $a_{n+1}$, at the beginning and the end of $S$, respectively ($S$ now has length $n+2$). This is to ensure that $S$ must start and end with a \mybox{$\heartsuit$}.

Finally, $P$ publicly appends a face-down ``marking card'' \mybox{$\diamondsuit$}, called $a_{n+2}$, at the end of $S$ ($S$ now has length $n+3$). This is to mark the beginning and the end of $S$ after $S$ has been shifted cyclically several times thoughout the protocol. See Fig. \ref{figM1}.

\begin{figure}[H]
\centering
\begin{tikzpicture}
\node at (0.0,0.5) {\mybox{$\heartsuit$}};
\node at (0.6,0.5) {\mybox{$\clubsuit$}};
\node at (1.2,0.5) {\mybox{$\clubsuit$}};
\node at (1.8,0.5) {\mybox{$\heartsuit$}};
\node at (2.4,0.5) {\mybox{$\heartsuit$}};
\node at (3.0,0.5) {\mybox{$\heartsuit$}};
\node at (3.6,0.5) {\mybox{$\clubsuit$}};
\node at (4.2,0.5) {\mybox{$\clubsuit$}};
\node at (4.8,0.5) {\mybox{$\heartsuit$}};
\node at (5.4,0.5) {\mybox{$\heartsuit$}};
\node at (6.0,0.5) {\mybox{$\clubsuit$}};
\node at (6.6,0.5) {\mybox{$\heartsuit$}};
\node at (7.2,0.5) {\mybox{$\diamondsuit$}};

\node at (0.0,0) {$a_0$};
\node at (0.6,0) {$a_1$};
\node at (1.2,0) {$a_2$};
\node at (1.8,0) {$a_3$};
\node at (2.4,0) {$a_4$};
\node at (3.0,0) {$a_5$};
\node at (3.6,0) {$a_6$};
\node at (4.2,0) {$a_7$};
\node at (4.8,0) {$a_8$};
\node at (5.4,0) {$a_9$};
\node at (6.05,0) {$a_{10}$};
\node at (6.7,0) {$a_{11}$};
\node at (7.3,0) {$a_{12}$};
\end{tikzpicture}
\caption{A sequence $S$ representing the third row of the solution in Fig \ref{fig0}}
\label{figM1}
\end{figure}

The verification is divided into the following three phases.

\subsection{Phase 1: Counting Blocks of Black Cells}
Currently, $S$ contains $k$ blocks of consecutive \mybox{$\clubsuit$}s. In this phase, $P$ will reveal the length of each block, then replace all \mybox{$\clubsuit$}s in $S$ with \mybox{$\spadesuit$}s.

$P$ performs the following steps for $k$ iterations. In the $i$-th iteration,

\begin{enumerate}
	\item Apply the chosen cut protocol to $S$ to select a card corresponding to the leftmost cell of the $i$-th leftmost block of black cells in $R$ (the block with length $x_i$). Let $a_j$ denote the selected card.
	\item Turn over cards $a_j,a_{j+1},a_{j+2},\dots,a_{j+x_i-1}$ (where the indices are taken modulo $n+3$) to reveal that they are all \mybox{$\clubsuit$}s. Otherwise, $V$ rejects.
	\item Turn over cards $a_{j-1}$ and $a_{j+x_i}$ (where the indices are taken modulo $n+3$) to reveal that they are both \mybox{$\heartsuit$}s. Otherwise, $V$ rejects.
	\item Replace every face-up \mybox{$\clubsuit$} with a face-up \mybox{$\spadesuit$}. This is to mark that this block of black cells has already been verified.
	\item Turn over all face-up cards.
\end{enumerate}

After $k$ iterations, $V$ is convinced that $R$ contains at least $k$ different blocks of black cells with lengths $x_1,x_2,\dots,x_k$, but does not know the order of these blocks, or whether $R$ contains any additional black cells besides the ones in these $k$ blocks. Also, all \mybox{$\clubsuit$}s in $S$ have been replaced with \mybox{$\spadesuit$}s. See Fig. \ref{figM2}.

\begin{figure}[H]
\centering
\begin{tikzpicture}
\node at (0.0,0.5) {\mybox{$\heartsuit$}};
\node at (0.6,0.5) {\mybox{$\spadesuit$}};
\node at (1.2,0.5) {\mybox{$\spadesuit$}};
\node at (1.8,0.5) {\mybox{$\heartsuit$}};
\node at (2.4,0.5) {\mybox{$\heartsuit$}};
\node at (3.0,0.5) {\mybox{$\heartsuit$}};
\node at (3.6,0.5) {\mybox{$\spadesuit$}};
\node at (4.2,0.5) {\mybox{$\spadesuit$}};
\node at (4.8,0.5) {\mybox{$\heartsuit$}};
\node at (5.4,0.5) {\mybox{$\heartsuit$}};
\node at (6.0,0.5) {\mybox{$\spadesuit$}};
\node at (6.6,0.5) {\mybox{$\heartsuit$}};
\node at (7.2,0.5) {\mybox{$\diamondsuit$}};
\end{tikzpicture}
\caption{The sequence $S$ from Fig. \ref{figM1} at the end of Phase 1 (in a cyclic rotation)}
\label{figM2}
\end{figure}

\subsection{Phase 2: Removing White Cells}
Currently, $S$ contains $k+1$ blocks of consecutive \mybox{$\heartsuit$}s (including a block at the beginning which contains $a_0$, and a block at the end which contains $a_{n+1}$). In this phase, $P$ will remove some \mybox{$\heartsuit$}s from $S$ such that there will be exactly one remaining \mybox{$\heartsuit$} in each block.

Let $X = x_1+x_2+\dots+x_k$. As there are $n-X$ white cells in $R$, there must be $n-X+2$ \mybox{$\heartsuit$}s in $S$ (including $a_0$ and $a_{n+1}$). $P$ performs the following steps for $(n-X+2)-(k+1) = n-X-k+1$ iterations.

\begin{enumerate}
	\item Apply the chosen cut protocol to $S$ to select any \mybox{$\heartsuit$} such that there are currently at least two remaining \mybox{$\heartsuit$}s in a block it belongs to.
	\item Turn over the selected card to reveal that it is a \mybox{$\heartsuit$}. Otherwise, $V$ rejects.
	\item Remove that card from $S$.
\end{enumerate}

After $n-X-k+1$ iterations, each pair of adjacent blocks of \mybox{$\spadesuit$}s in $S$ are now separated by exactly one \hbox{\mybox{$\heartsuit$},} and there is also a \mybox{$\heartsuit$} before the first block and after the last block ($S$ now has length $X+k+2$). See Fig. \ref{figM3}.

\begin{figure}[H]
\centering
\begin{tikzpicture}
\node at (0.0,0.5) {\mybox{$\heartsuit$}};
\node at (0.6,0.5) {\mybox{$\spadesuit$}};
\node at (1.2,0.5) {\mybox{$\spadesuit$}};
\node at (1.8,0.5) {\mybox{$\heartsuit$}};
\node at (2.4,0.5) {\mybox{$\spadesuit$}};
\node at (3.0,0.5) {\mybox{$\spadesuit$}};
\node at (3.6,0.5) {\mybox{$\heartsuit$}};
\node at (4.2,0.5) {\mybox{$\spadesuit$}};
\node at (4.8,0.5) {\mybox{$\heartsuit$}};
\node at (5.4,0.5) {\mybox{$\diamondsuit$}};
\end{tikzpicture}
\caption{The sequence $S$ from Fig. \ref{figM1} at the end of Phase 2 (in a cyclic rotation)}
\label{figM3}
\end{figure}

\subsection{Phase 3: Verifying Order of Blocks of Black Cells}
$P$ applies the random cut to $S$, turns over all cards, and shifts the sequence cyclically such that the rightmost card is a \mybox{$\diamondsuit$}.

$V$ verifies that the remaining cards in $S$ consist of one \mybox{$\heartsuit$}, $x_1$ consecutive \hbox{\mybox{$\spadesuit$}s}, one \mybox{$\heartsuit$}, $x_2$ consecutive \hbox{\mybox{$\spadesuit$}s}, \dots, one \mybox{$\heartsuit$}, $x_k$ consecutive \hbox{\mybox{$\spadesuit$}s}, one \mybox{$\heartsuit$}, and one \mybox{$\diamondsuit$} in this order from left to right. Otherwise, $V$ rejects.

$P$ performs the above three phases of verification for every row and column of the grid. If all rows and columns pass the verification, then $V$ accepts.

\subsection{Optimization}
As $P$ only uses one card per cell in the verification of a row and a column it belongs to, a total of two cards per cell are actually used in our protocol. Therefore, duplicating a sequence on each cell at the beginning is not necessary. Instead, if $P$ applies the copy protocol in Section \ref{copy} without putting cards in the third row of $M$ in Step 1(b), the protocol will just verify that the input sequence is in a correct format (either \hbox{\mybox{$\clubsuit$}\mybox{$\heartsuit$}} or \hbox{\mybox{$\heartsuit$}\mybox{$\clubsuit$}}) in Step 3, and will return the input sequence in the second row of $M$ in Step 4. This modified copy protocol uses the same idea as the one developed by Mizuki and Shizuya \cite{verify}.

After verifying that a sequence on each cell is in a correct format, $P$ uses the left card of the sequence to verify a row, and the right card to verify a column the cell belongs to. When verifying a column, the selected card will be a \mybox{$\heartsuit$} if the cell is black and a \mybox{$\clubsuit$} if the cell is white, so we have to treat \mybox{$\clubsuit$} and \mybox{$\heartsuit$} exactly the opposite way throughout the protocol.

After the optimization, our protocol uses $mn+1$ \mybox{$\clubsuit$}s, $mn+\max(m,n)+4$ \mybox{$\heartsuit$}s, $\max(m,n)$ \mybox{$\spadesuit$}s, and one \mybox{$\diamondsuit$}, resulting in a total of $2mn+2\max(m,n)+6 = \Theta(mn)$ cards. The protocol also uses $mn+2m+2n+2w = \Theta(mn)$ shuffles, where $w$ is the total number of white cells in the grid.

\section{Security Proof of Protocol for Nonogram}
We will prove the perfect completeness, perfect soundness, and zero-knowledge properties of the protocol for Nonogram.

\begin{lemma}[Perfect Completeness] \label{lem1}
If $P$ knows a solution of the Nonogram puzzle, then $V$ always accepts.
\end{lemma}

\begin{proof}
Assume that $P$ knows a solution. Consider the verification of any row $R$.

In each $i$-th iteration during Phase 1, $P$ selects from $S$ a card $a_j$ corresponding to the leftmost cell of the $i$-th leftmost block of black cells in $R$. As that block has length $x_i$ and has never been selected before, the cards $a_j,a_{j+1},a_{j+2},\dots,a_{j+x_i-1}$ must all be \mybox{$\clubsuit$}s, so Step 2 will pass. Also, since there is at least one white cell between two adjacent blocks of black cells (and at least one \mybox{$\heartsuit$} to the left of the leftmost block of \mybox{$\clubsuit$}s and to the right of the rightmost block of \mybox{$\clubsuit$}s), both $a_{j-1}$ and $a_{j+x_i}$ must be \mybox{$\heartsuit$}s, so Step 3 will pass. Thus, Phase 1 of the verification will pass.

At the start of Phase 2, $S$ contains exactly $k+1$ blocks of \mybox{$\heartsuit$}s, which together have a total of $n-X+2$ \mybox{$\heartsuit$}s. In each iteration, $P$ removes one \mybox{$\heartsuit$} from $S$ such that each block still has at least one \mybox{$\heartsuit$}. $P$ can do so as many as $(n-X+2)-(k+1) = n-X-k+1$ times, so Step 2 will pass for all $n-X-k+1$ iterations. Moreover, at the end of Phase 2, there will be exactly $k+1$ remaining \mybox{$\heartsuit$}s in $S$, which means each block contains exactly one \mybox{$\heartsuit$}.

At the start of Phase 3, there is exactly one \mybox{$\heartsuit$} between two adjacent blocks of \hbox{\mybox{$\spadesuit$}s} in $S$ (and also a \mybox{$\heartsuit$} at the beginning and the end of $S$). Also, the blocks of \hbox{\mybox{$\spadesuit$}s} in $S$ are arranged in the same order as the corresponding blocks of black cells in $R$, so the lengths of these blocks must be $x_1,x_2,\dots,x_k$ in this order from left to right. Thus, Phase 3 of the verification will pass.

As the proof holds for the verification of every row (and also of every column analogously), we can conclude that $V$ always accepts.
\end{proof}

\begin{lemma}[Perfect Soundness] \label{lem2}
If $P$ does not know a solution of the Nonogram puzzle, then $V$ always rejects.
\end{lemma}

\begin{proof}
We will prove the contrapositive of this statement. Assume that $V$ accepts, which means the verification of every row and column passes. Consider the verification of any row $R$.

In each $i$-th iteration during Phase 1, the steps $P$ performs ensure that there exists a block of exactly $x_i$ consecutive black cells in $R$. As all \hbox{\mybox{$\clubsuit$}s} in the blocks $P$ has selected in previous iterations have already been replaced with \mybox{$\spadesuit$}s, this block must be different from the blocks $P$ selected in previous iterations. Thus, $R$ must contain at least $k$ different blocks of black cells with lengths $x_1,x_2,\dots,x_k$ (in some order).

Also, only \mybox{$\heartsuit$}s are removed from $S$ during Phase 2, and there is no remaining \mybox{$\clubsuit$} in $S$ during Phase 3. This implies $R$ contains no other black cells besides the ones in these $k$ blocks.

Furthermore, in Phase 3, the lengths of the blocks of \mybox{$\spadesuit$}s in $S$ are $x_1,x_2,\dots,$ $x_k$ in this order from left to right. Since the blocks of \mybox{$\spadesuit$}s in $S$ are arranged in the same order as the blocks of black cells in $R$, $R$ must contains exactly $k$ blocks of consecutive black cells with lengths $x_1,x_2,\dots,x_k$ in this order from left to right.

As the proof holds for the verification of every row (and also of every column analogously), we can conclude that $P$ knows a solution of the Nonogram puzzle.
\end{proof}

\begin{lemma}[Zero-Knowledge] \label{lem3}
During the verification, $V$ does not obtain any information about $P$'s solution of the Nonogram puzzle.
\end{lemma}

\begin{proof}
To prove the zero-knowledge property, we will construct a simulator $S$ that does not know $P$'s solution, but can simulate all distributions of values that are revealed when cards are turned face-up.

\begin{itemize}
	\item Consider Step 3 of the copy protocol in Section \ref{copy} where cards are turned face-up. The revealed sequence has probability $1/2$ to be each of \hbox{\mybox{$\clubsuit$}\mybox{$\heartsuit$}} and \hbox{\mybox{$\heartsuit$}\mybox{$\clubsuit$}} due to the pile-shifting shuffle in Step 2. Therefore, this step can be simulated by $S$ without knowing $P$'s solution.
	\item Consider Step 3 of the chosen cut protocol in Section \ref{chosen} where cards are turned face-up. The only \mybox{$\clubsuit$} has probability $1/k$ to be at each of the $k$ positions due to the pile-shifting shuffle in Step 2. Therefore, this step can be simulated by $S$ without knowing $P$'s solution.
	\item Consider the verification of each row (resp. column) in the main protocol. There is only one deterministic pattern of cards that are turned face-up in all phases. This pattern solely depends on the sequence $(x_1,x_2,\dots,x_k)$ assigned to that row (resp. column), which is public information, so the whole protocol can be simulated by $S$ without knowing $P$'s solution.
\end{itemize}
\end{proof}

\section{Application to Nonogram Color}
The idea to verify a solution of Nonogram Color is similar to that of Nonogram. However, there are two main issues we have to consider and make modifications to our protocol.

First, in Phase 1, when verifying a block of $x_i$ consecutive cells with the $q_i$-th color, the cells right next to the left and the right of this block may not be white, but can be any color that is not the $q_i$-th color. In particular, $P$ cannot reveal the colors of these two cells (which will leak information about the solution to $V$) but have to show $V$ that they do not have the $q_i$-th color.

Suppose there are $p$ colors used in the puzzle (including white, which is denoted as the first color). For $1 \leq q \leq p$, we define $E_p(q)$ to be a sequence of $p$ consecutive cards where all of them being \mybox{$\heartsuit$}s except the $q$-th leftmost card being a \mybox{$\clubsuit$} (e.g. $E_4(2)$ is \mybox{$\heartsuit$}\mybox{$\clubsuit$}\mybox{$\heartsuit$}\mybox{$\heartsuit$}). We use $E_p(q)$ to encode a cell with the $q$-th color. By encoding the colors this way, $P$ can reveal only the $q$-th card of the sequence to show $V$ that the corresponding cell does not have the $q$-th color without revealing its actual color.

Second, in Phase 2, when removing white cells, we cannot leave one white cell between each adjacent blocks of painted cells (because some adjacent blocks may be right next to each other and do not have any white cell between them to begin with), so we have to remove every white cell from the row in Phase 2. To avoid having adjacent blocks with the same color merge with each other to become one large block, when marking each block as verified in Phase 1, we have to also mark the length of that block.

In the modified protocol, besides \mybox{$\clubsuit$}s and \mybox{$\heartsuit$}s, we also use cards with a number on the front sides (all cards still have indistinguishable back sides). Define $i \circ E_p(q)$ to be a sequence of $p+1$ cards consisting of a card \mybox{$i$} concatenated by $E_p(q)$ (e.g. $3 \circ E_4(2)$ is \hbox{\mybox{3}\mybox{$\heartsuit$}\mybox{$\clubsuit$}\mybox{$\heartsuit$}\mybox{$\heartsuit$}}). In Phase 1, originally a cell with the $q_i$-th color is encoded by $0 \circ E_p(q_i)$. After verifying a block of $x_i$ consecutive cells with the $q_i$-th color, $P$ replaces a sequence on every verified cell with $x_i \circ E_p(q_i)$ (similar to marking with a \mybox{$\spadesuit$} in the original protocol). By marking the cells this way, two adjacent blocks of, say, two green cells and three green cells (with at least one white cell between them) will not be mistakenly interpreted as a single block of five green cells, even after the white cells between them are removed.

\subsection{Modified Subprotocols}
To support the modified protocol, the following two subprotocols can also be applied to a sequence of $k$ stacks of cards (instead of a sequence of $k$ cards), as long as every stack has an equal number of cards.

For the random cut protocol in Section \ref{rc}, if we have a sequence of $k$ stacks with each having $\ell$ cards (instead of a sequence of $k$ cards), we can implement the protocol in exactly the same way as the pile-shifing shuffle in Section \ref{pss} on an $\ell \times k$ matrix.

For the chosen cut protocol in Section \ref{chosen}, if we have a sequence $A=(a_1,a_2,\dots,a_k)$ of $k$ stacks with each having $\ell$ cards (instead of a sequence of $k$ cards), $P$ can implement the protocol in exactly the same way to select a stack $a_i$ he/she desires without revealing $i$ to $V$.

\subsection{Generalized Copy Protocol} \label{copy2}
The following protocol is a generalized version of the copy protocol in Section \ref{copy}.

Given a sequence $E_p(q)$ for some $1 \leq q \leq p$, a \textit{generalized copy protocol} \cite{polygon} enables $P$ to produce an additional copy of the input sequence without revealing the value of $q$ to $V$. It also verifies that the input sequence is in the form $E_p(q)$ for some $1 \leq q \leq p$.

\begin{figure}[H]
\centering
\begin{tikzpicture}
\node at (0.0,1.4) {\mybox{?}};
\node at (0.6,1.4) {\mybox{?}};
\node at (1.2,1.4) {\mybox{?}};
\node at (1.8,1.4) {\dots};
\node at (2.4,1.4) {\mybox{?}};
\node at (3.0,1.4) {\mybox{?}};

\node at (0.0,0.7) {\mybox{$\heartsuit$}};
\node at (0.6,0.7) {\mybox{$\heartsuit$}};
\node at (1.2,0.7) {\mybox{$\heartsuit$}};
\node at (1.8,0.7) {\dots};
\node at (2.4,0.7) {\mybox{$\heartsuit$}};
\node at (3.0,0.7) {\mybox{$\clubsuit$}};

\node at (0.0,0) {\mybox{$\heartsuit$}};
\node at (0.6,0) {\mybox{$\heartsuit$}};
\node at (1.2,0) {\mybox{$\heartsuit$}};
\node at (1.8,0) {\dots};
\node at (2.4,0) {\mybox{$\heartsuit$}};
\node at (3.0,0) {\mybox{$\clubsuit$}};
\end{tikzpicture}
\caption{A $3 \times p$ matrix constructed in Step 2 of the generalized copy protocol}
\label{fig22}
\end{figure}

In the generalized copy protocol, $P$ performs the followingn steps.
\begin{enumerate}
	\item Reverse the input sequence, i.e. make each $i$-th leftmost card become the $i$-th rightmost card. Note that this reversed sequence is $E_p(p+1-q)$.
	\item Construct the following $3 \times p$ matrix $M$ (see Fig. \ref{fig22}).
	\begin{enumerate}
		\item In the first row, place the reversed input sequence obtained from Step 1.
		\item In the second row and third row, publicly place a face-up sequence $E_p(p)$.
	\end{enumerate}
	\item Turn over all face-up cards and apply the pile-shifting shuffle to $M$.
	\item Turn over all cards in the first row of $M$. Shift the columns of $M$ cyclically such that the only \mybox{$\clubsuit$} in the first row moves to the leftmost column.
	\item The sequences in the second and third rows of $M$ will be the two copies of the input sequence as desired.
\end{enumerate}

\section{Protocol for Nonogram Color}
On each cell in the Nonogram Color grid, $P$ secretly places a face-down sequence $E_p(q)$ if the cell has the $q$-th color according to $P$'s solution (recall that white is the first color and thus a white cell is encoded by $E_p(1)$). Then, $P$ publicly applies the generalized copy protocol to the sequence on each cell to produce an additional copy of it. Each of the two copies will be used to verify a row and a column the cell belongs to. Note that the generalized copy protocol also verifies that the sequence on each cell is in a correct format ($E_p(q)$ for some $1 \leq q \leq p$).

From now on, we will show the verification of a row $R$ with $n$ cells that has a sequence $(x_1,x_2,\dots,x_k)$ assigned to it, with each number $x_i$ having the $q_i$-th color. The verification of a column works analogously.

For each cell in $R$, $P$ picks one copy of the sequence on it and stack the cards in that sequence into a single stack (with the leftmost card being the topmost card in the stack), then publicly puts a face-down \mybox{0} on top of the stack (the stack now has $p+1$ cards). $P$ does this for every cell in $R$ to form a sequence of $n$ stacks $S = (a_0,a_1,\dots,a_{n-1})$, where each stack in $S$ corresponding to each cell in $R$ in this order from left to right.

Finally, $P$ publicly appends a face-down ``marking stack'' $-1 \circ E_p(1)$, called $a_n$, at the end of $S$ ($S$ now has length $n+1$). This marking stack functions exactly like a marking card \mybox{$\diamondsuit$} in the original protocol for Nonogram. See Fig. \ref{figN1}.

\begin{figure}[H]
\centering
\begin{tikzpicture}
\node at (0,0.5) {$\mystack{?}$};
\node at (1,0.5) {$\mystack{?}$};
\node at (2,0.5) {$\mystack{?}$};
\node at (3,0.5) {$\mystack{?}$};
\node at (4,0.5) {$\mystack{?}$};
\node at (5,0.5) {$\mystack{?}$};
\node at (6,0.5) {$\mystack{?}$};
\node at (7,0.5) {$\mystack{?}$};
\node at (8,0.5) {$\mystack{?}$};
\node at (9,0.5) {$\mystack{?}$};
\node at (10,0.5) {$\mystack{?}$};

\node at (0,0) {$0 \circ$};
\node at (1,0) {$0 \circ$};
\node at (2,0) {$0 \circ$};
\node at (3,0) {$0 \circ$};
\node at (4,0) {$0 \circ$};
\node at (5,0) {$0 \circ$};
\node at (6,0) {$0 \circ$};
\node at (7,0) {$0 \circ$};
\node at (8,0) {$0 \circ$};
\node at (9,0) {$0 \circ$};
\node at (10,0) {$-1 \circ$};

\node at (0,-0.3) {$E_4(2)$};
\node at (1,-0.3) {$E_4(2)$};
\node at (2,-0.3) {$E_4(2)$};
\node at (3,-0.3) {$E_4(3)$};
\node at (4,-0.3) {$E_4(3)$};
\node at (5,-0.3) {$E_4(3)$};
\node at (6,-0.3) {$E_4(1)$};
\node at (7,-0.3) {$E_4(1)$};
\node at (8,-0.3) {$E_4(1)$};
\node at (9,-0.3) {$E_4(2)$};
\node at (10,-0.3) {$E_4(1)$};
\end{tikzpicture}
\caption{A sequence $S$ representing the fifth row of the solution in Fig \ref{fig00} (White, red, green, and blue are denoted as the first, the second, the third, and the fourth colors, respectively.)}
\label{figN1}
\end{figure}

The verification is divided into the following three phases.

\subsection{Phase 1: Counting Blocks of Painted Cells}
Currently, $R$ contains $k$ blocks of consecutive painted cells. In this phase, $P$ will reveal the length of each block as well as marking it with its length.

$P$ performs the following steps for $k$ iterations. In the $i$-th itration,

\begin{enumerate}
	\item Apply the chosen cut protocol to $S$ to select a stack corresponding to the leftmost cell of the $i$-th leftmost block of painted cells in $R$ (the block with the $q_i$-th color and length $x_i$). Let $a_j$ denote the selected stack.
	\item Turn over all cards in stacks $a_j,a_{j+1},a_{j+2},\dots,a_{j+x_i-1}$ (where the indices are taken modulo $n+1$) to reveal that every stack is $0 \circ E_p(q_i)$, i.e. each corresponding cell has the $q_i$-th color. Otherwise, $V$ rejects.
	\item Turn over only the $(q_i+1)$-th cards of stacks $a_{j-1}$ and $a_{j+x_i}$ (where the indices are taken modulo $n+1$) to reveal that they are both \mybox{$\heartsuit$}s, i.e. the two corresponding cells do not have the $q_i$-th color. Otherwise, $V$ rejects.
	\item In each of the stacks $a_j,a_{j+1},a_{j+2},\dots,a_{j+x_i-1}$, replace the topmost card \mybox{0} with a \mybox{$x_i$}. Turn over all face-up cards. The purpose of this step is to mark that this block has been verified and has length $x_i$.
\end{enumerate}

After $k$ iterations, $V$ is convinced that $R$ contains at least $k$ different blocks of painted cells with lengths $x_1,x_2,\dots,x_k$, but does not know the order of these blocks, or whether $R$ contains any additional painted cells besides the ones in these $k$ blocks. Also, all \mybox{0}s in the corresponding blocks in $S$ have been replaced with cards with positive numbers. See Fig. \ref{figN2}.

\begin{figure}[H]
\centering
\begin{tikzpicture}
\node at (0,0.5) {$\mystack{?}$};
\node at (1,0.5) {$\mystack{?}$};
\node at (2,0.5) {$\mystack{?}$};
\node at (3,0.5) {$\mystack{?}$};
\node at (4,0.5) {$\mystack{?}$};
\node at (5,0.5) {$\mystack{?}$};
\node at (6,0.5) {$\mystack{?}$};
\node at (7,0.5) {$\mystack{?}$};
\node at (8,0.5) {$\mystack{?}$};
\node at (9,0.5) {$\mystack{?}$};
\node at (10,0.5) {$\mystack{?}$};

\node at (0,0) {$3 \circ$};
\node at (1,0) {$3 \circ$};
\node at (2,0) {$3 \circ$};
\node at (3,0) {$3 \circ$};
\node at (4,0) {$3 \circ$};
\node at (5,0) {$3 \circ$};
\node at (6,0) {$0 \circ$};
\node at (7,0) {$0 \circ$};
\node at (8,0) {$0 \circ$};
\node at (9,0) {$1 \circ$};
\node at (10,0) {$-1 \circ$};

\node at (0,-0.3) {$E_4(2)$};
\node at (1,-0.3) {$E_4(2)$};
\node at (2,-0.3) {$E_4(2)$};
\node at (3,-0.3) {$E_4(3)$};
\node at (4,-0.3) {$E_4(3)$};
\node at (5,-0.3) {$E_4(3)$};
\node at (6,-0.3) {$E_4(1)$};
\node at (7,-0.3) {$E_4(1)$};
\node at (8,-0.3) {$E_4(1)$};
\node at (9,-0.3) {$E_4(2)$};
\node at (10,-0.3) {$E_4(1)$};
\end{tikzpicture}
\caption{The sequence $S$ from Fig. \ref{figN1} at the end of Phase 1 (in a cyclic rotation)}
\label{figN2}
\end{figure}

\subsection{Phase 2: Removing White Cells}
In this phase, $P$ will remove all stacks of $0 \circ E_p(1)$ (which correspond to white cells) from $S$.

Let $X = x_1+x_2+\dots+x_k$. $P$ performs the following steps for $n-X$ iterations.

\begin{enumerate}
	\item Apply the chosen cut protocol to $S$ to select any stack of $0 \circ E_p(1)$.
	\item Turn over all cards in that stack to reveal that it is a $0 \circ E_p(1)$. Otherwise, $V$ rejects.
	\item Remove that stack from $S$.
\end{enumerate}

After $n-X$ iterations, all stacks of $0 \circ E_p(1)$ have been removed from $S$ ($S$ now has length $X+1$). See Fig. \ref{figN3}.

\begin{figure}[H]
\centering
\begin{tikzpicture}
\node at (0,0.5) {$\mystack{?}$};
\node at (1,0.5) {$\mystack{?}$};
\node at (2,0.5) {$\mystack{?}$};
\node at (3,0.5) {$\mystack{?}$};
\node at (4,0.5) {$\mystack{?}$};
\node at (5,0.5) {$\mystack{?}$};
\node at (6,0.5) {$\mystack{?}$};
\node at (7,0.5) {$\mystack{?}$};

\node at (0,0) {$3 \circ$};
\node at (1,0) {$3 \circ$};
\node at (2,0) {$3 \circ$};
\node at (3,0) {$3 \circ$};
\node at (4,0) {$3 \circ$};
\node at (5,0) {$3 \circ$};
\node at (6,0) {$1 \circ$};
\node at (7,0) {$-1 \circ$};

\node at (0,-0.3) {$E_4(2)$};
\node at (1,-0.3) {$E_4(2)$};
\node at (2,-0.3) {$E_4(2)$};
\node at (3,-0.3) {$E_4(3)$};
\node at (4,-0.3) {$E_4(3)$};
\node at (5,-0.3) {$E_4(3)$};
\node at (6,-0.3) {$E_4(2)$};
\node at (7,-0.3) {$E_4(1)$};
\end{tikzpicture}
\caption{The sequence $S$ from Fig. \ref{figN1} at the end of Phase 2 (in a cyclic rotation)}
\label{figN3}
\end{figure}

\subsection{Phase 3: Verifying Order of Blocks of Painted Cells}
$P$ applies the random cut to $S$, turns over all cards in all stacks, and shifts the sequence cyclically such that the rightmost stack is $-1 \circ E_p(1)$.

$V$ verifies that the remaining stacks in $S$ are: $x_1$ stacks of $x_1 \circ E_p(q_1)$, $x_2$ stacks of $x_2 \circ E_p(q_2)$, \dots, $x_k$ stacks of $x_k \circ E_p(q_k)$, and one stack of $-1 \circ E_p(1)$ in this order from left to right. Otherwise, $V$ rejects.

$P$ performs the above three phases of verification for every row and column of the grid. If all rows and columns pass the verification, then $V$ accepts.

The modified protocol for Nonogram Color uses $\Theta(mnp)$ cards and $\Theta(mn)$ shuffles.

\section{Security Proof of Protocol for Nonogram Color}
The proofs of perfect completeness, perfect soundness, and zero-knowledge properties of the modified protocol for Nonogram Color are very similar to those of the original protocol for Nonogram. For the sake of completeness, the full proofs are shown in this section.

\begin{lemma}[Perfect Completeness] \label{lem4}
If $P$ knows a solution of the Nonogram Color puzzle, then $V$ always accepts.
\end{lemma}

\begin{proof}
Assume that $P$ knows a solution. Consider the verification of any row $R$.

In each $i$-th iteration during Phase 1, $P$ selects from $S$ a stack $a_j$ corresponding to the leftmost cell of the $i$-th leftmost block of painted cells in $R$. As that block has the $q_i$-th color with length $x_i$, and has never been selected before, the stacks $a_j,a_{j+1},a_{j+2},\dots,a_{j+x_i-1}$ must all be stacks of $0 \circ E_p(q_i)$, so Step 2 will pass. Also, since the cells next to the left and right of this block must have colors different from the $q_i$-th color, the $(q_i+1)$-th cards of stacks $a_{j-1}$ and $a_{j+x_i}$ must both be \hbox{\mybox{$\heartsuit$}s}, so Step 3 will pass. Thus, Phase 1 of the verification will pass.

As $R$ contains exactly $n-X$ white cells, at the start of Phase 2 $S$ contains exactly $n-X$ stacks of $0 \circ E_p(1)$. In each iteration, $P$ removes one $0 \circ E_p(1)$ from $S$. $P$ can do so as many as $n-X$ times, so Step 2 will pass for all $n-X$ iterations. Moreover, at the end of Phase 2, there will be no stack of $0 \circ E_p(1)$ left in $S$.

At the start of Phase 3, there is no stack of $0 \circ E_p(1)$ left in $S$. Also, the blocks of stacks of $x_i \circ E_p(q_i)$ in $S$ are arranged in the same order as the corresponding blocks of painted cells in $R$, so $S$ must consist of blocks of $x_1,x_2,\dots,x_k$ consecutive stacks of $x_1 \circ E_p(q_1), x_2 \circ E_p(q_2), \dots, x_k \circ E_p(q_k)$ in this order from left to right. Thus, Phase 3 of the verification will pass.

As the proof holds for the verification of every row (and also of every column analogously), we can conclude that $V$ always accepts.
\end{proof}

\begin{lemma}[Perfect Soundness] \label{lem5}
If $P$ does not know a solution of the Nonogram Color puzzle, then $V$ always rejects.
\end{lemma}

\begin{proof}
We will prove the contrapositive of this statement. Assume that $V$ accepts, which means the verification of every row and column passes. Consider the verification of any row $R$.

In each $i$-th iteration during Phase 1, the steps $P$ performs ensure that there exists a block of exactly $x_i$ consecutive cells with the $q_i$-th color in $R$. As the topmost \mybox{0}s on the stacks $P$ has selected in previous iterations have already been replaced with cards with positive numbers, this block must be different from the blocks $P$ selected in previous iterations. Thus, $R$ must contain at least $k$ different blocks of painted cells with the $q_1,q_2,\dots,q_k$-th colors (in some order) and lengths $x_1,x_2,\dots,x_k$ (in the same order).

Also, only stacks of $0 \circ E_p(1)$ are removed from $S$ during Phase 2, and there is no remaining stack of $0 \circ E_p(q)$ for $q>1$ in $S$ during Phase 3. This implies $R$ contains no other painted cells besides the ones in these $k$ blocks.

Furthermore, in Phase 3, the lengths of the blocks of \mybox{$\spadesuit$}s in $S$ are $x_1,x_2,\dots,$ $x_k$ in this order from left to right. Since the blocks of \mybox{$\spadesuit$}s in $S$ are arranged in the same order as the blocks of black cells in $R$, $R$ must contains exactly $k$ blocks of consecutive painted cells with the $q_1,q_2,\dots,q_k$-th colors and lengths $x_1,x_2,\dots,x_k$ in this order from left to right.

As the proof holds for the verification of every row (and also of every column analogously), we can conclude that $P$ knows a solution of the Nonogram Color puzzle.
\end{proof}

\begin{lemma}[Zero-Knowledge] \label{lem6}
During the verification, $V$ does not obtain any information about $P$'s solution of the Nonogram Color puzzle.
\end{lemma}

\begin{proof}
To prove the zero-knowledge property, we will construct a simulator $S$ that does not know $P$'s solution, but can simulate all distributions of values that are revealed when cards are turned face-up.

\begin{itemize}
	\item Consider Step 4 of the generalized copy protocol in Section \ref{copy2} where cards are turned face-up. The only \mybox{$\clubsuit$} has probability $1/p$ to be at each of the $p$ positions due to the pile-shifting shuffle in Step 3. Therefore, this step can be simulated by $S$ without knowing $P$'s solution.
	\item Consider Step 3 of the chosen cut protocol in Section \ref{chosen} where cards are turned face-up. The only \mybox{$\clubsuit$} has probability $1/k$ to be at each of the $k$ positions due to the pile-shifting shuffle in Step 2. Therefore, this step can be simulated by $S$ without knowing $P$'s solution.
	\item Consider the verification of each row (resp. column) in the main protocol for Nonogram Color. There is only one deterministic pattern of cards that are turned face-up in all phases. This pattern solely depends on the sequence $(x_1,x_2,\dots,x_k)$ assigned to that row (resp. column) and the colors $q_1,q_2,\dots,q_k$ of the numbers in that sequence, which is public information, so the whole protocol can be simulated by $S$ without knowing $P$'s solution.
\end{itemize}
\end{proof}

\section{Future Work}
We constructed card-based ZKP protocols for Nonogram using $\Theta(mn)$ cards and $\Theta(mn)$ shuffles, and for Nonogram Color using $\Theta(mnp)$ cards and $\Theta(mn)$ shuffles. A possible future work is to improve ZKP protocols for these two puzzles so that they can be implemented using a deck containing all different cards with no duplicates like the ones for Sudoku in \cite{sudoku2} and Makaro in \cite{makaro2}. Other challenging future work includes developing card-based ZKP protocols for other popular pencil puzzles and improving practicalness (type of cards, number of cards, or number of shuffles) of the existing protocols.


\begin{thebibliography}{99}
	\bibitem{akari} X. Bultel, J. Dreier, J.-G. Dumas and P. Lafourcade. Physical Zero-Knowledge Proofs for Akari, Takuzu, Kakuro and KenKen. In \textit{Proceedings of the 8th International Conference on Fun with Algorithms (FUN)}, pp. 8:1--8:20 (2016).
	\bibitem{makaro} X. Bultel, J. Dreier, J.-G. Dumas, P. Lafourcade, D. Miyahara, T. Mizuki, A. Nagao, T. Sasaki, K. Shinagawa and H. Sone. Physical Zero-Knowledge Proof for Makaro. In \textit{Proceedings of the 20th International Symposium on Stabilization, Safety, and Security of Distributed Systems (SSS)}, pp. 111--125 (2018).
	\bibitem{nonogram} Y.-F. Chien and W.-K. Hon. Cryptographic and Physical Zero-Knowledge Proof: From Sudoku to Nonogram. In \textit{Proceedings of the 5th International Conference on Fun with Algorithms (FUN)}, pp. 102--112 (2010).
	\bibitem{norinori} J.-G. Dumas, P. Lafourcade, D. Miyahara, T. Mizuki, T. Sasaki and H. Sone. Interactive Physical Zero-Knowledge Proof for Norinori. In \textit{Proceedings of the 25th International Computing and Combinatorics Conference (COCOON)}, pp. 166--177 (2019).
	\bibitem{abc} T. Fukusawa and Y. Manabe. Card-Based Zero-Knowledge Proof for the Nearest Neighbor Property: Zero-Knowledge Proof of ABC End View. In \textit{Proceedings of the 12th International Conference on Security, Privacy and Applied Cryptographic Engineering (SPACE)}, pp. 147--161 (2022).
	\bibitem{zkp} O. Goldreich, S. Micali and A. Wigderson. Proofs that yield nothing but their validity and a methodology of cryptographic protocol design. \textit{Journal of the ACM}, 38(3): 691--729 (1991).
	\bibitem{zkp0} S. Goldwasser, S. Micali and C. Rackoff. The knowledge complexity of interactive proof systems. \textit{SIAM Journal on Computing}, 18(1): 186--208 (1989).
	\bibitem{google} Google Play: Nonogram. \url{https://play.google.com/store/search?q=Nonogram&c=apps}
	\bibitem{google2} Google Play: Nonogram Color. \url{https://play.google.com/store/search?q=Nonogram\%20Color&c=apps}
	\bibitem{sudoku0} R. Gradwohl, M. Naor, B. Pinkas and G.N. Rothblum. Cryptographic and Physical Zero-Knowledge Proof Systems for Solutions of Sudoku Puzzles. \textit{Theory of Computing Systems}, 44(2): 245--268 (2009).
	\bibitem{koch} A. Koch and S. Walzer. Foundations for Actively Secure Card-Based Cryptography. In \textit{Proceedings of the 10th International Conference on Fun with Algorithms (FUN)}, pp. 17:1--17:23 (2020).
	\bibitem{slitherlink} P. Lafourcade, D. Miyahara, T. Mizuki, L. Robert, T. Sasaki and H. Sone. How to construct physical zero-knowledge proofs for puzzles with a ``single loop'' condition. \textit{Theoretical Computer Science}, 888: 41--55 (2021).
	\bibitem{takuzu} D. Miyahara, L. Robert, P. Lafourcade, S. Takeshige, T. Mizuki, K. Shinagawa, A. Nagao and H. Sone. Card-Based ZKP Protocols for Takuzu and Juosan. In \textit{Proceedings of the 10th International Conference on Fun with Algorithms (FUN)}, pp. 20:1--20:21 (2020).
	\bibitem{kakuro} D. Miyahara, T. Sasaki, T. Mizuki and H. Sone. Card-Based Physical Zero-Knowledge Proof for Kakuro. \textit{IEICE Transactions on Fundamentals of Electronics, Communications and Computer Sciences}, E102.A(9): 1072--1078 (2019).
	\bibitem{verify} T. Mizuki and H. Shizuya. Practical Card-Based Cryptography. In \textit{Proceedings of the 7th International Conference on Fun with Algorithms (FUN)}, pp. 313--324 (2014).
	\bibitem{mizuki} T. Mizuki and H. Sone. Six-Card Secure AND and Four-Card Secure XOR. In \textit{Proceedings of the 3rd International Frontiers of Algorithmics Workshop (FAW)}, pp. 358--369 (2009).
	\bibitem{suguru} L. Robert, D. Miyahara, P. Lafourcade, L. Libralesso and T. Mizuki. Physical zero-knowledge proof and NP-completeness proof of Suguru puzzle. \textit{Information and Computation}, 285(B): 104858 (2022).
	\bibitem{nurikabe} L. Robert, D. Miyahara, P. Lafourcade and T. Mizuki. Card-Based ZKP for Connectivity: Applications to Nurikabe, Hitori, and Heyawake. \textit{New Generation Computing}, 40(1): 149--171 (2022).
	\bibitem{nurimisaki} L. Robert, D. Miyahara, P. Lafourcade and T. Mizuki. Card-Based ZKP Protocol for Nurimisaki. In \textit{Proceedings of the 24th International Symposium on Stabilization, Safety, and Security of Distributed Systems (SSS)}, pp. 285--298 (2022).
	\bibitem{usowan} L. Robert, D. Miyahara, P. Lafourcade and T. Mizuki. Hide a Liar: Card-Based ZKP Protocol for Usowan. In \textit{Proceedings of the 17th Annual Conference on Theory and Applications of Models of Computation (TAMC)}, pp. 201--217 (2022).
	\bibitem{nonogram2} S. Ruangwises. An Improved Physical ZKP for Nonogram. In \textit{Proceedings of the 15th Annual International Conference on Combinatorial Optimization and Applications (COCOA)}, pp. 262--272 (2021).
	\bibitem{sudoku2} S. Ruangwises. Two Standard Decks of Playing Cards are Sufficient for a ZKP for Sudoku. \textit{New Generation Computing}, 40(1): 49--65 (2022).
	\bibitem{shikaku} S. Ruangwises and T. Itoh. How to Physically Verify a Rectangle in a Grid: A Physical ZKP for Shikaku. In \textit{Proceedings of the 11th International Conference on Fun with Algorithms (FUN)}, pp. 22:1--22:12 (2022).
	\bibitem{numberlink} S. Ruangwises and T. Itoh. Physical Zero-Knowledge Proof for Numberlink Puzzle and $k$ Vertex-Disjoint Paths Problem. \textit{New Generation Computing}, 39(1): 3--17 (2021).
	\bibitem{ripple} S. Ruangwises and T. Itoh. Physical Zero-Knowledge Proof for Ripple Effect. \textit{Theoretical Computer Science}, 895: 115--123 (2021).
	\bibitem{bridges} S. Ruangwises and T. Itoh. Physical ZKP for Connected Spanning Subgraph: Applications to Bridges Puzzle and Other Problems. In \textit{Proceedings of the 19th International Conference on Unconventional Computation and Natural Computation (UCNC)}, pp. 149--163 (2021).
	\bibitem{makaro2} S. Ruangwises and T. Itoh. Physical ZKP for Makaro Using a Standard Deck of Cards. In \textit{Proceedings of the 17th Annual Conference on Theory and Applications of Models of Computation (TAMC)}, pp. 43--54 (2022).
	\bibitem{sudoku} T. Sasaki, D. Miyahara, T. Mizuki and H. Sone. Efficient card-based zero-knowledge proof for Sudoku. \textit{Theoretical Computer Science}, 839: 135--142 (2020).
	\bibitem{polygon} K. Shinagawa, T. Mizuki, J.C.N. Schuldt, K. Nuida, N. Kanayama, T. Nishide, G. Hanaoka and E. Okamoto. Card-Based Protocols Using Regular Polygon Cards. \textit{IEICE Transactions on Fundamentals of Electronics, Communications and Computer Sciences}, E100.A(9): 1900--1909 (2017).
	\bibitem{hindu} I. Ueda, D. Miyahara, A. Nishimura, Y. Hayashi, T. Mizuki and H. Sone. Secure implementations of a random bisection cut. \textit{International Journal of Information Security}, 19(4): 445--452 (2020).
	\bibitem{np} N. Ueda and T. Nagao. NP-completeness Results for NONOGRAM via Parsimonious Reductions. Technical Report TR96-0008, Department of Computer Science, Tokyo Institute of Technology (1996).
\end{thebibliography}
\end{document}